%% file: arxiv.tex
\documentclass[11pt]{article}

\usepackage[margin=1in]{geometry}

\usepackage[english]{babel}
\usepackage[colorinlistoftodos]{todonotes}
\renewcommand{\epsilon}{\varepsilon}

\usepackage[colorlinks=true,citecolor=blue,urlcolor=blue]{hyperref}

\usepackage{amsmath}
\usepackage{thmtools}

\usepackage{hyperref}

\usepackage{amsthm}
\usepackage{amssymb}
\usepackage{amsfonts}
\usepackage{dsfont}
\usepackage{thm-restate}
\usepackage{nicefrac}
\usepackage{footnote}
\usepackage{comment}
\usepackage{caption}
\usepackage{subcaption}
\usepackage{mdframed}

\usepackage[capitalize]{cleveref}
\usepackage{enumerate}

\usepackage{algorithm}
\usepackage{algorithmic}

\usepackage{tikzsymbols}
\usepackage{tikz}

\tikzstyle{vertex}=[circle, draw,fill=gray!30, inner sep=0pt, minimum size=16pt]
\tikzstyle{svertex}=[circle, draw,fill=gray!30, inner sep=0pt, minimum size=10pt]
\tikzstyle{sgvertex}=[circle, draw,fill=gray!15, inner sep=0pt, minimum size=12pt]
\tikzstyle{ssgvertex}=[circle, draw,fill=gray!15, inner sep=0pt, minimum size=6pt]
\tikzstyle{sssgvertex}=[circle, draw,fill=gray!15, inner sep=0pt, minimum size=4pt]
\tikzstyle{sssagvertex}=[circle, draw,fill=gray!15, inner sep=0pt, minimum size=3.2pt]
\tikzstyle{ssssgvertex}=[circle, draw,fill=gray!15, inner sep=0pt, minimum size=2pt]

\newtheorem{theorem}{Theorem}
\newtheorem{proposition}[theorem]{Proposition}
\newtheorem{lemma}[theorem]{Lemma}

\newtheorem{conjecture}[theorem]{Conjecture}

\theoremstyle{definition}
\newtheorem{definition}[theorem]{Definition}

\newtheorem{remark}[theorem]{Remark}

\DeclareMathOperator{\spn}{span}
\DeclareMathOperator{\rank}{rank}
\DeclareMathOperator{\polylog}{polylog}

\newcommand{\OPT}{\ensuremath{\textrm{OPT}}}
\newcommand{\E}{\mathbb{E}}

\begin{document}

\title{Semi-Streaming Algorithms for Submodular Matroid Intersection\footnote{This research was supported by the Swiss National Science Foundation project
200021-184656 “Randomness in Problem Instances and Randomized Algorithms.”}}
 \author{
 Paritosh Garg\\
   EPFL\\
  \texttt{paritosh.garg@epfl.ch}
     \and
   Linus Jordan\\
   EPFL\\
  \texttt{linus.jordan@bluewin.ch}
     \and
   Ola Svensson\\
  EPFL\\
   \texttt{ola.svensson@epfl.ch}
 }
\date{}

\maketitle

%
%

%
%
%
%

\begin{abstract}

While the basic greedy algorithm gives a semi-streaming algorithm with an approximation guarantee of $2$ for the \emph{unweighted} matching problem, it was only recently that Paz and Schwartzman obtained an analogous result for weighted instances. Their approach  is based on the versatile local ratio technique and also applies to generalizations such as weighted hypergraph matchings. However, the framework for the analysis fails for the  related problem of weighted matroid intersection and as a result the approximation guarantee for weighted instances did not match the factor $2$  achieved by the greedy algorithm for unweighted instances.Our main result closes this gap by developing a semi-streaming algorithm with an approximation guarantee of $2+\epsilon$ for \emph{weighted} matroid intersection, improving upon the previous best guarantee of $4+\epsilon$. 
Our techniques also allow us to generalize recent results by Levin and Wajc on submodular maximization subject to matching constraints to that of matroid-intersection constraints. 

While our algorithm is an adaptation of the local ratio technique used in previous works, the analysis deviates significantly and relies on structural properties of matroid intersection, called kernels. Finally, we also conjecture that our algorithm gives a $(k+\epsilon)$ approximation for the intersection of $k$ matroids but prove that new tools are needed in the analysis  as the used structural properties fail for $k\geq 3$.



\end{abstract}

\section{Introduction}



For large problems, it is often not realistic that the entire input can be stored in random access memory so more memory efficient algorithms are preferable.  A popular model for such algorithms is the (semi-)streaming  model (see e.g.~\cite{Muthu}): the elements of the input are fed to the algorithm in a stream and the algorithm is required to have a small memory footprint. 

Consider the classic maximum matching problem in an undirected graph $G=(V,E)$.  An algorithm in the semi-streaming model\footnote{This model can also be considered in the multi-pass setting when the algorithm is allowed to take several passes over the stream. However, in this work we focus on the most basic and widely studied setting in which the algorithm takes a single pass over the stream.} is fed the edges one-by-one in a stream $e_1, e_2, \ldots,$ $e_{|E|}$ and at any point of time the algorithm is only allowed $O(|V|\polylog(|V|))$ bits of storage.  The goal is to output a large matching $M \subseteq E$ at the end of the stream. Note that the allowed memory usage is sufficient for the algorithm to store a solution $M$ but in general it is much smaller than the size of the input since the number of edges may be as many as  $|V|^2/2$. Indeed, the intuitive difficulty in designing a semi-streaming algorithm is that the algorithm needs to discard many of the seen edges (due to the memory restriction) without knowing the future edges and still return a good solution at the end of the stream.  

For the unweighted matching problem, the best known semi-streaming algorithm is the basic greedy approach:
    \begin{center}
\begin{minipage}{0.95\textwidth}
\begin{mdframed}[hidealllines=true, backgroundcolor=gray!15]
    Initially, let $M= \emptyset$. Then for each edge $e$ in the stream, add it to $M$ if $M\cup \{e\}$ is a feasible solution, i.e., a matching; otherwise the edge $e$ is discarded.
\end{mdframed}
\end{minipage}
\end{center}
The algorithm uses space $O(|V| \log |V|)$ and a simple proof shows that it returns a $2$-approximate solution in the \emph{unweighted} case, i.e, a matching of size at least half the size of an maximum matching.  However, this basic approach fails to achieve any approximation guarantee for \emph{weighted graphs}. 

Indeed, for weighted matchings, it is non-trivial to even get a small constant-factor approximation. One way to do so is to replace edges if we have a much heavier edge. This is formalized in \cite{Feigenbaum} who get a $6$-approximation. Later, \cite{McGregor} improved this algorithm to find a $5.828$-approximation; and, with a more involved technique, \cite{Crouch} provided a $(4+\epsilon)$-approximation. 

It was only in recent breakthrough work~\cite{Paz} that the gap in the approximation guarantee between unweighted and weighted matchings was closed. Specifically, \cite{Paz} gave a semi-streaming algorithm for weighted matchings with an approximation guarantee of $2+\epsilon$ for every $\epsilon>0$. Shortly after,  \cite{Wajc} came up with a simplified analysis of their algorithm, reducing the memory requirement from $O_\epsilon(|V|\log
^2 |V|)$ to $O_\epsilon(|V|\log |V|)$. These results for weighted matchings are tight (up  to the $\epsilon$) in the sense that any improvement would also improve the state-of-the-art in the unweighted case, which is a long-standing open problem.

The algorithm of~\cite{Paz} is an elegant use of the local ratio technique~\cite{lrsurvey, lrtheorem} in the semi-streaming setting.  While this technique is very versatile and it readily generalizes to weighted hypergraph matchings, it is much harder to use it for the related problem of weighted matroid intersection. 
This is perhaps surprising as  many of the prior results for the matching problem  also applies to the matroid intersection problem in the semi-streaming model (see \cref{sec:prelim} for definitions). Indeed, the greedy algorithm still returns a $2$-approximate solution in the unweighted case and the algorithm in~\cite{Crouch} returns a $(4+\epsilon)$-approximate solution for weighted instances. 
So, prior to our work, the status of the matroid intersection problem was that of the matching problem \emph{before}~\cite{Paz}. 

We now describe on a high-level the reason that the techniques from~\cite{Paz} are not easily applicable to matroid intersection and our approach for dealing with this difficulty. The approach in~\cite{Paz} works  in two parts, first certain elements of the stream are selected and added to a set $S$, and then at the end of the stream a matching $M$ is computed by the greedy algorithm that inspects the edges of $S$ in the reverse order in which they were added. This way of constructing the solution $M$ greedily by going backwards in time is a standard framework for analyzing algorithms based on the local ratio technique.
Now in order to adapt their algorithm to matroid intersection, recall that the bipartite matching problem can be formulated as the intersection of two partition matroids. We can thus reinterpret their algorithm and analysis in this setting. Furthermore, after this reinterpretation, it is not too hard to define an algorithm that works for the intersection of any two matroids. However, bipartite matching is a \emph{special} case of matroid intersection which captures a rich set of seemingly more complex problems. This added expressiveness causes the analysis and the standard framework for analyzing local ratio algorithms to fail. Specifically, we prove that a solution formed by running the greedy algorithm on $S$ in the reverse order (as done for the matching problem) fails to give any constant-factor approximation guarantee for the matroid intersection problem. To overcome this and to obtain our main result,  we make a connection to a concept called matroid kernels  (see \cite{Fleiner} for more details about kernels), which allows us to, in a more complex way, identify a subset of $S$ with an approximation guarantee of $2+\epsilon$.

 Finally, for the intersection of more than two matroids, the same approach in the analysis  does not work, because the notion of matroid kernel does not generalize to more than two matroids. However, we conjecture that the subset $S$ generated for the intersection of $k$ matroids still contains a $(k+\epsilon)$-approximation. Currently, the best approximation results are a $(k^2+\epsilon)$-approximation from \cite{Crouch} and  a $(2(k+\sqrt{k(k-1)})-1)$-approximation from \cite{Chak}. For $k=3$, the former is better, giving a $(9+\epsilon)$-approximation. For $k>3$, the latter is better, giving an $O(k)$-approximation.

 \paragraph{Generalization to submodular functions.} Very recently, Levin and Wajc~\cite{levin2020streaming}  obtained improved approximation ratios for matching and b-matching problems in the semi-streaming model with respect to submodular functions. 
 Specifically, they get a $(3+2\sqrt{2})$-approximation for monotone submodular b-matching, $(4+3\sqrt{2})$-approximation for non-monotone submodular matching, and a $(3+\epsilon)$-approximation for maximum weight (linear) b-matching. In our paper, we are able to extend our algorithm for weighted matroid intersection to work with submodular functions by combining our and their ideas. In fact, we are able to generalize all their results to the case of matroid intersection with better or equal\footnote{One can get rid of the $\delta$ factor if we assume that the function value is polynomially bounded by $|E|$, an assumption made by \cite{levin2020streaming}.} approximation ratios: we get  $(3+2\sqrt{2}+\delta)$-approximation for monotone submodular matroid intersection, $(4+3\sqrt{2}+\delta)$-approximation for non-monotone submodular matroid intersection and $(2+\epsilon)$-approximation for maximum weight (linear) matroid intersection.

  \paragraph{Outline.} In \cref{sec:prelim} we introduce basic matroid concepts and we formally define the weighted matroid intersection problem in the semi-streaming model. \cref{sec:mainalg} and \cref{sec:algmemefficient} are devoted to our main result, i.e., the semi-streaming algorithm for weighted matroid intersection with an approximation guarantee of $(2+\varepsilon)$. Specifically, in \cref{sec:mainalg} we adapt the algorithm of~\cite{Paz} without worrying about the memory requirements, show why the standard analysis fails, and then give our new analysis. We then make the obtained algorithm memory efficient in \cref{sec:algmemefficient}. Further in \cref{submodular_section}, we adapt our algorithm to work with submodular functions by using ideas from \cite{levin2020streaming}. Finally, in \cref{morethantwo}, we discuss the case of  more than two matroids.

\section{Preliminaries}
\label{sec:prelim}

\paragraph{Matroids.} We define and  give a brief overview of the basic concepts related to matroids that we use in this paper. For a more comprehensive treatment, we refer the reader to~\cite{schrijver2003combinatorial}. A \emph{matroid} is a tuple $M= (E,I)$ consisting of a finite ground set $E$ and a family  $I \subseteq 2^E$  of subsets of $E$ satisfying:
\begin{itemize}
    \item if $X\subseteq Y,Y\in I$, then $X\in I$; and
    \item if $X\in I,Y\in I$ and $|Y|>|X|$, then $\exists \ e\in Y\setminus X$ such that $X\cup \{e\}\in I$.
\end{itemize}
The elements in $I$ (that are subsets of $E$) are referred to as the \emph{independent sets} of the matroid and the set $E$ is referred to as the \emph{ground set}.  With a matroid $M = (E,I)$, we associate the \emph{rank function} $\rank_M : 2^E \rightarrow \mathbb{N}$ and the \emph{span function} $\spn_M: 2^E \rightarrow 2^E$ defined as follows for every $E' \subseteq E$,
\begin{align*}
    \rank_M(E') & = \max \{ |X| \mid X \subseteq E' \mbox{ and } X \in I\},  \\
    \spn_M(E') & = \{e \in E \mid \rank_M(E' \cup \{e\}) = \rank_M(E')\}.
\end{align*}
We simply write $\rank(\cdot)$ and $\spn(\cdot)$ when the matroid $M$ is clear from the context.
In words, the rank function equals the size of the largest independent set when restricted to $E'$ and the span function equals the elements in $E'$ and all elements that cannot be added to a maximum cardinality independent set of $E'$ while maintaining independence. 
The \emph{rank of the matroid} equals $\rank(E)$, i.e., the size of the largest independent set.

\paragraph{The weighted matroid intersection problem in the semi-streaming model.}
In the \emph{weighted matroid intersection problem}, we are given two matroids $M_1 = (E, I_1), M_2 = (E, I_2)$ on a common ground set $E$ and a non-negative weight function $w: E \rightarrow \mathbb{R}_{\geq 0}$ on the elements of the ground set. 
The goal is to find a subset $X \subseteq E$ that is independent in both matroids, i.e., $X\in I_1$ and $X\in I_2$, and whose weight $w(X) = \sum_{e\in X} w(e)$ is maximized. 

In seminal work~\cite{EDMONDS197939}, Edmonds gave a polynomial-time algorithm for solving the weighted matroid intersection problem to optimality in the classic model of computation when the whole input is available to the algorithm throughout the computation. In contrast, the problem becomes significantly harder and tight results are still eluding us in the semi-streaming model where the memory footprint of the algorithm  and its access pattern to the input are restricted.   
Specifically,  in the \emph{semi-streaming} model the ground set $E$ is revealed in a stream $e_1, e_2, \ldots, e_{|E|}$ and at time $i$ the algorithm gets access to $e_i$ and can perform computation based on $e_i$ and its current memory but without knowledge of future elements $e_{i+1}, \ldots, e_{|E|}$. The algorithm has independence-oracle access to the matroids $M_1$ and $M_2$ restricted to the elements stored in the memory, i.e., for a set of such elements, the algorithm can query  whether the set is independent in each matroid.. The goal is to design an algorithm such that (i) the memory usage is near-linear $O((r_1 + r_2) \polylog(r_1 + r_2))$ at any time, where $r_1$ and $r_2$ denote the ranks of the input matroids $M_1$ and $M_2$, respectively,  and (ii) at the end of the stream the algorithm should output a feasible solution $X\subseteq E$, i.e., a subset $X$ that satisfies $X \in I_1$ and $X \in I_2$, of large weight $w(X)$. We remark that the memory requirement $O((r_1 + r_2) \polylog(r_1 + r_2))$ is natural as $r_1 + r_2 = |V|$ when formulating a bipartite matching problem as the intersection of two matroids\footnote{The considered problem can also be formulated as the problem of finding an independent set in one matroid the matroids, say $M_1$, and maximizing a submodular function which would be the (weighted) rank function of $M_2$. For that problem,~\cite{DBLP:journals/corr/abs-2002-05477} recently gave a streaming algorithm with an approximation guarantee of $(2+\epsilon)$. However, the space requirement of their algorithm is exponential the rank of $M_1$ (which would correspond to be exponential in $|V|$ in the matching case) and thus it does not provide a meaningful algorithm for our setting.}.

The difficulty in designing a good semi-streaming algorithm is that the memory requirement is much smaller than the size of the ground set $E$ and thus the algorithm must intuitively discard many of the elements without knowledge of the future and without significantly deteriorating the weight of the final solution $X$. The quality of the algorithm is measured in terms of its approximation guarantee: an algorithm is said to have an \emph{approximation guarantee} of $\alpha$ if it is guaranteed to output a solution $X$, no matter the input and the order of the stream, such that $w(X) \geq \OPT/\alpha$ where $\OPT$ denotes the weight of an optimal solution to the instance. As aforementioned, our main result in this paper is a semi-streaming algorithm with an approximation guarantee of $2+\varepsilon$, for every $\varepsilon>0$, improving upon the previous best guarantee of $4+\epsilon$~\cite{Crouch}.

\section{The Local Ratio Technique for Weighted Matroid Intersection}
\label{sec:mainalg}
In this section, we first present the local ratio algorithm for the weighted matching problem that forms the basis of the semi-streaming algorithm in~\cite{Paz}.  
We then adapt it to the weighted matroid intersection problem. 
While the algorithm is fairly natural to adapt to this setting, we give an example in \cref{sec:counterexample} that shows that the same techniques as used for analyzing the algorithm for matchings does not work for matroid intersection. 
Instead, our analysis, which is presented in \cref{sec:mlr2analysis}, deviates from the standard framework for analyzing local ratio algorithms and it heavily relies on a structural property of matroid intersection known as kernels.   We remark that the algorithms considered in this section do not have a small memory footprint. We deal with this in \cref{sec:algmemefficient} to obtain our semi-streaming algorithm.

\subsection{Local-Ratio Technique for Weighted Matching}
The local ratio algorithm for the weighted matching problem is given in \cref{lr}. The algorithm maintains vertex potentials $w(u)$ for every vertex $u$, a set $S$ of selected edges, and an auxiliary weight function $g: S \rightarrow \mathbb{R}_{\geq 0}$ of the selected edges. Initially the vertex potentials are set to $0$ and the set $S$ is empty. When an edge $e=\{u,v\}$ arrives, the algorithm computes how much it gains compared to the previous edges, by taking its weight minus the weight/potential of its endpoints ($g(e)=w(e)-w(u)-w(v)$). If the gain is positive, then we add the edge to $S$, and add the gain to the weight of the endpoints, that is, we set $w(u)=w(u)+g(e)$ and $w(v)=w(v)+g(e)$.
\begin{figure}[t]
    \centering
    \input{LRExampleFigure}
    \caption{The top part shows an example execution of the local ratio technique for weighted matchings (\cref{lr}). The bottom part shows how to adapt this (bipartite) example to the language of weighted matroid intersection (\cref{mlr2}).}
    \label{fig:simpleLR}
\end{figure}
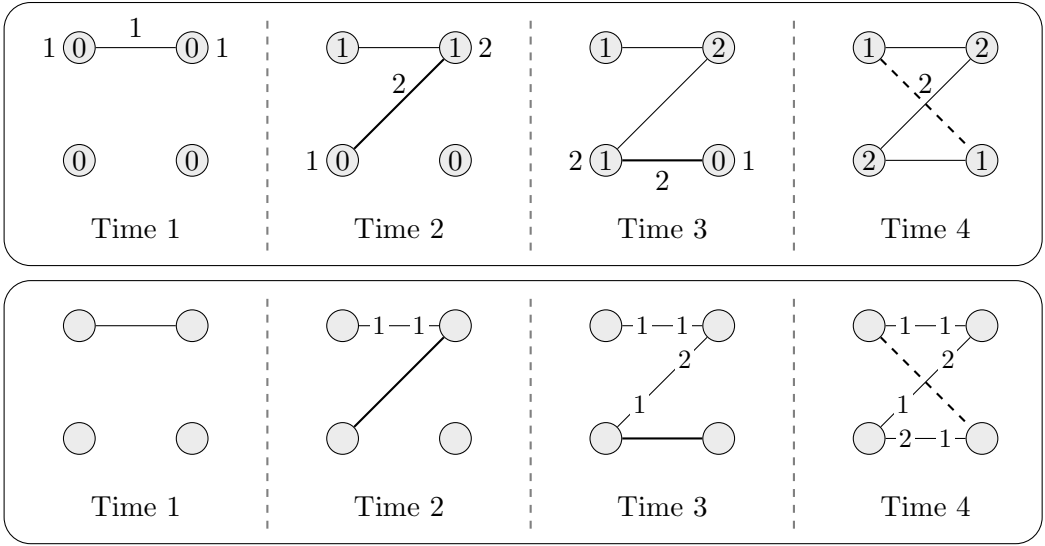
\begin{algorithm}[H] 
\caption{Local ratio algorithm for weighted matching} 
\label{lr} 
\begin{algorithmic}[1] 
    \REQUIRE A stream of the edges of a graph $G=(V,E)$ with a weight function $w:E\rightarrow \mathbb{R}_{\geq 0}$. 
    \ENSURE A matching $M$.
    \STATE $S \leftarrow \emptyset$
    \STATE $\forall u\in V, w(u)\leftarrow 0$
    \FOR{edge $e=(u,v)$ in the stream}
    \IF{$w(u)+w(v)<w(e)$}
    \STATE $g(e)\leftarrow w(e)-w(u)-w(v)$
    \STATE $w(u)\leftarrow w(u)+g(e)$
    \STATE $w(v)\leftarrow w(v)+g(e)$
    \STATE $S \leftarrow S \cup \{e\}$
    \ENDIF
    \ENDFOR
    \RETURN a maximum weight matching $M$ among the edges stored on the stack $S$
\end{algorithmic}
\end{algorithm}

For a better intuition of the algorithm, consider the example depicted on the top of Figure~\ref{fig:simpleLR}. The stream consists of four edges $e_1, e_2, e_3, e_4$ with weights $w(e_1) =1$ and $w(e_2) = w(e_3) = w(e_4) = 2$. At each time step $i$, we depict the arriving edge $e_i$ in thick along with its weight; the vertex potentials before the algorithm considers this edge is written on the vertices, and the updated vertex potentials (if any) after considering $e_i$ are depicted next to the incident vertices. The edges that are added to $S$ are solid  and those that are not added to $S$ are dashed.  

At the arrival of the first edge of weight $w(e_1) = 1$, both incident vertices have potential $0$ and so the algorithm adds this edge to $S$ and increases the incident vertex potentials with the gain $g(e_1) = 1$. For the second edge of weight $w(e_2) = 2$, the sum of incident vertex potentials is $1$ and so the gain of $e_2$ is $g(e_2) = 2 - 1$, which in turn causes the algorithm to add this edge to $S$ and to increase the incident vertex potentials by $1$. The third time step is similar to the second. At the last time step, edge $e_4$ arrives of weight $w(e_4) = 2$. As the incident vertex potentials sum up to $2$ the gain of $e_4$ is not strictly positive and so this edge is \emph{not} added to $S$ and no vertex potentials are updated.   Finally, the algorithm returns the maximum weight matching in $S$ which in this case consists of edges $\{e_1, e_3\}$  and has weight $3$. Note that the optimal matching of this instance had weight $4$ and we thus found a $4/3$-approximate solution. 

In general, the algorithm has an approximation guarantee of $2$. This is proved using a common framework to analyze algorithms based on the local ratio technique: We ignore the weights and greedily  construct a matching $M$ by inspecting the edges in $S$ in reverse order, i.e., we first consider the edges that were added last. An easy proof (see e.g.~\cite{Wajc}) then shows that the  matching $M$  constructed in this way has weight at least half the optimum weight. 

In the next section, we adapt the above described algorithm to the context of matroid intersections. We also give an example that the above framework for the analysis fails to give any constant-factor approximation guarantee. Our alternative (tight) analysis of this algorithm is then given in \cref{sec:mlr2analysis}.
\subsection{Adaptation to Weighted Matroid Intersection}
\begin{algorithm}[t] 
\caption{Local ratio for matroid intersection} 
\label{mlr2} 
\begin{algorithmic} 
    \REQUIRE A stream of the elements of the common ground set  of matroids $M_1 = (E, I_1), M_2 = (E, I_2)$.
    \ENSURE A set $X\subseteq E$ that is independent in both matroids.
    \STATE $S \leftarrow \emptyset$
    \FOR{element $e$ in the stream}
    \STATE calculate $w^*_i(e) =\max\left( \{0\} \cup \lbrace \theta: e \in \spn_{M_i}\left(\lbrace f\in S\ |\ w_i(f)\geq \theta\rbrace\right)\rbrace \right)$ for $i\in \{1,2\}$.
    \IF{$w(e)>w_1^*(e)+w^*_2(e)$}
    \STATE $g(e)\leftarrow w(e)-w_1^*(e)-w^*_2(e)$
    \STATE $w_1(e)\leftarrow w^*_1(e)+g(e)$
    \STATE $w_2(e)\leftarrow w^*_2(e)+g(e)$
    \STATE $S \leftarrow S \cup \{e\}$
    \ENDIF
    \ENDFOR
    \RETURN a maximum weight set $T\subseteq S$ that is independent in $M_1$ and $M_2$
    
\end{algorithmic}
\end{algorithm}
%
When adapting \cref{lr} to matroid intersection to obtain \cref{mlr2}, the first problem we encounter is the fact that matroids do not have a notion of vertices, so we cannot keep a weight/potential for each vertex. To describe how we overcome this issue, it is helpful to consider the case of bipartite matching and in particular the example depicted in Figure~\ref{fig:simpleLR}. It is well known that the weighted matching problem on a bipartite graph with edge set $E$ and bipartition $V_1, V_2$ can be modelled as a weighted matroid intersection problem on matroids $M_1 = (E, I_1)$ and $M_2 = (E, I_2)$ where  for $i\in \{1,2\}$
\begin{align*}
    I_i  = \{E' \subseteq E \mid \mbox{each vertex $v\in V_i$ is incident to at most one vertex in $E'$}\}\,.
\end{align*}

Instead of keeping a weight for each vertex, we will maintain two weight functions $w_1$ and $w_2$, one for each matroid. These weight functions will be set so that the following holds in the special case of bipartite matching:  on the arrival of a new edge $e$, let $T_i \subseteq S$ be an independent set in $I_i$ of selected edges  that maximizes the weight function $w_i$. Then we have that
\begin{align}
    \min_{f\in T_i: T_i \setminus \{f\} \cup \{e\}\in I_i} w_i(f) \qquad \mbox{if $T_i \cup \{e\} \not \in I_i$ and $0$ otherwise}
    \label{eq:adapt1}
\end{align}
equals the vertex potential of the incident vertex $V_i$ when running \cref{lr}. It is well-known (e.g. by the optimality of the greedy algorithm for matroids) that the cheapest element $f$ to remove from $T_i$ to make $T_i \setminus \{f\} \cup \{e\}$ an independent set equals the largest weight $\theta$ so that the elements of weight at least $\theta$ spans $e$. We thus have that~\eqref{eq:adapt1}  equals
\begin{align*}
     \max\left( \{0\} \cup \lbrace \theta: e \in \spn_{M_i}\left(\lbrace f\in S\ |\ w_i(f)\geq \theta\rbrace\right)\rbrace \right)
\end{align*}
and it follows that the quantities $w_1^*(e)$ and $w_2^*(e)$ in \cref{mlr2} equal the incident vertex potentials in  $V_1$ and $V_2$ of \cref{lr} in the special case of bipartite matching. To see this, let us return to our example in  Figure~\ref{fig:simpleLR} and let $V_1$ be the two vertices on the left and $V_2$ be the two vertices on the right. In the bottom part of the figure, the weight functions $w_1$ and $w_2$ are depicted (at the corresponding side of the edge) after the arrival of each edge. At time step $1$, $e_1$ does not need to replace any elements in any of the matroids and so $w^*_1(e_1) = w^*_1(e_2) = 0$. We therefore have that its gain is $g(e_1) = 1$ and the algorithm  sets $w_1(e_1) = w_2(e_1) = 1$.  At time $2$, edge $e_2$ of weight $2$ arrives. It is not spanned in the first matroid  whereas it is spanned by  edge $e_1$ of weight $1$ in the second matroid. It follows that $w_1^*(e_2) = 0$ and $w_2^*(e_2) = w_2(e_1) = 1$ and so $e_2$ has positive gain $g(e_2) = 1$ and it sets $w_1(e_2) = 1$ and $w_2(e_2) = w_2(e_1) + 1 = 2$. The third time step is similar to the second. At the last time step, $e_4$ of weight $2$ arrives. However, since it is spanned by  $e_1$ with $w_1(e_1) = 1$ in the first matroid and by $e_3$ with $w_2(e_3) = 1$  in the second matroid, its gain is $0$ and it is thus not added to the set $S$. Note that throughout this example, and in general for bipartite graphs, \cref{mlr2} is identical to \cref{lr}. One may therefore expect that the analysis of \cref{lr} also generalizes to \cref{mlr2}. We explain next that this is not the case for general matroids.

\subsubsection{Counter Example to Same Approach in Analysis}
\label{sec:counterexample}
We give a simple example showing that the greedy selection (as done in the analysis for \cref{lr} for weighted matching) does not work for matroid intersection. Still, it turns out that the set $S$ generated by \cref{mlr2} always contains a 2-approximation  but the selection process is more involved.
\begin{lemma}
There exist two matroids $M_1=(E,I_1)$ and $M_2=(E,I_2)$ on a common ground set $E$ and a weight function $w: E \rightarrow \mathbb{R}_{\geq 0}$ such that a greedy algorithm that considers the elements in the set $S$ in the reverse order of when they were added by \cref{mlr2}   does not provide any constant-factor approximation.
\end{lemma}
\begin{proof}
    The example consists of the ground set $E=\lbrace a,b,c,d\rbrace$ with weights $w(a)=1,w(b)=1+\epsilon,w(c)=2\epsilon,w(d)=3\epsilon$ for a small $\epsilon>0$ (the approximation guarantee will be at least $\Omega(1/\varepsilon)$). The matroids $M_1 = (E, I_1)$ and $M_2 = (E, I_2)$ are defined by 
    \begin{itemize}
        \item a subset of $E$ is in $I_1$ if and only if it does not contain $\lbrace a,b\rbrace$; and
        \item a subset  of $E$ is in $I_2$ if and only if it contains at most two elements.
    \end{itemize}

To see that $M_1$ and $M_2$ are matroids, note that $M_1$ is a partition matroid with partitions $\lbrace a,b\rbrace,\lbrace c\rbrace,\lbrace d\rbrace$, and $M_2$ is the 2-uniform matroid (alternatively, one can easily check that $M_1$ and $M_2$ satisfy the definition of a matroid).

Now consider the execution of \cref{mlr2} when given the elements of $E$ in the order $a,b,c,d$:
\begin{itemize}
    \item Element $a$ has weight $1$, and $\lbrace a\rbrace$ is independent both in $M_1$ and $M_2$, so we set $w_1(a)=w_2(a)=g(a)=1$ and $a$ is added to $S$.
    \item Element $b$ is spanned by $a$ in $M_1$ and not spanned by any element in $M_2$. So we get $g(b)=w(b)-w_1^*(b)-w_2^*(b)=1+\epsilon-1-0=\epsilon$. As $\epsilon>0$, we add $b$ to $S$, and set $w_1(b)=w_1(a)+\epsilon = 1+\epsilon$ and $w_2(b)=\epsilon$.
    \item Element $c$ is not spanned by any element in $M_1$ but is spanned by $\{a,b\}$ in $M_2$. As $b$ has the smallest $w_2$ weight, $w_2^*(c) = w_2(b) = \epsilon$. So we have $g(c)=2\epsilon-w_1^*(c) - w_2^*(c) = 2\epsilon - 0 -\epsilon=\epsilon>0$, and we  set $w_1(c) = \epsilon$ and $w_2(c) = 2\epsilon$ and add $c$ to $S$. 
    \item Element $d$ is similar to $c$. We have $g(d)=3\epsilon-0-2\epsilon=\epsilon>0$ and so we set $w_1(d) = \epsilon$ and $w_2(d) = 3\epsilon$ and add $d$ to $S$. 
\end{itemize}
 As the algorithm selected all the elements, we have $S=E$. It follows that  the greedy algorithm on $S$ (in the reverse order of when elements were added) will select $d$ and $c$, after which the set is a maximal independent set in $M_2$. This gives a weight of $5\epsilon$, even though $a$ and $b$ both have weight at least 1, which shows that this algorithm does not guarantee any constant factor approximation.
\end{proof}

\subsection{Analysis of \cref{mlr2}}
\label{sec:mlr2analysis}

We prove that \cref{mlr2} has an approximation guarantee of $2$. 

\begin{theorem}
\label{k2}
Let $S$ be the subset generated by \cref{mlr2} on a stream $E$ of elements, matroids $M_1 = (E, I_1), M_2 = (E, I_2)$ and weight function $w: E \rightarrow \mathbb{R}_{\geq 0}$. Then there exists a subset $T\subseteq S$ independent in $M_1$ and in $M_2$ whose weight $w(T)$ is at least $w(S^*)/2$, where $S^*$ denotes an optimal solution to the weighted matroid intersection problem.
\end{theorem}

Throughout the analysis we fix the input matroids $M_1=(E, I_1), M_2=(E, I_2)$, the weight function $w: R \rightarrow \mathbb{R}_{\geq 0}$, and the order of the elements in the stream. While \cref{mlr2} only defines the weight functions $w_1$ and $w_2$ for the elements added to the set $S$, we extend them  in the analysis  by, for $i\in \{1,2\}$, letting $w_i(e) = w_i^*(e)$ for the elements $e$ not added to $S$. 

We now prove \cref{k2} by showing that $g(S) \geq w(S^*)/2$ (\cref{gee}) and that there is a solution $T\subseteq S$ such that $w(T) \geq g(S)$ (\cref{lemma:atleastgain}). In the proof of both these lemmas, we use the following properties of the computed set $S$.
\begin{lemma}\label{goodsubset}
    Let $S$ be the set generated by \cref{mlr2} and  $S'\subseteq S$ any subset. Consider one of the matroids $M_i$ with $i\in \{1,2\}$. There exists a subset $T'\subseteq S'$ that is independent in $M_i$, i.e.,  $T' \in I_i$, and $w_i(T') \geq g(S')$. Furthermore, the maximum weight independent set in $M_i$ over the whole ground set $E$ can be selected  to be a subset of $S$, i.e. $T_i \subseteq S$, and it satisfies $w_i(T_i)=g(S)$.
\end{lemma}

\begin{proof}
    Consider matroid $M_1$ (the proof is identical for $M_2$) and fix $S' \subseteq S$. The set $T_1' \subseteq S'$ that is independent in $M_1$ and that maximizes $w_1(T_1')$ satisfies
    \begin{align*}
        w_1(T'_1)= \int_0^\infty \rank(\{e \in T_1'\mid w_1(e) \geq \theta\})\,d\theta = \int_0^\infty \rank(\{e \in S'\mid w_1(e) \geq \theta\})\,d\theta \,.
    \end{align*}
    The second equality follows from the fact that the  greedy algorithm that considers the elements in decreasing order of weight is optimal for matroids and thus we have $\rank(\{e\in T_1'\mid w_1(e) \geq \theta\}) = \rank(\{e\in S'\mid w_1(e) \geq \theta\})$ for any $\theta\in \mathbb{R}$.
    
    Now index the elements of $S' = \{e_{1}, e_2, \ldots, e_\ell\}$ in the order they were added to $S$ by \cref{mlr2} and let $S'_j = \{e_1, \ldots, e_j\}$ for $j=0,1, \ldots, \ell$ (where $S'_0 = \emptyset$). By the above equalities and by telescoping,  
    \begin{align*}
         w_1(T'_1) & = \sum_{i=1}^\ell   \int_0^\infty  \left(\rank(\{e \in S_i'\mid w_1(e) \geq \theta\})- \rank(\{e \in S_{i-1}'\mid w_1(e) \geq \theta\})\right)\,d\theta\,.
    \end{align*}
    We have that $\rank(\{e \in S_i'\mid w_1(e) \geq \theta\})- \rank(\{e \in S_{i-1}'\mid w_1(e) \geq \theta\})$ equals $1$ if $w(e_i) \geq \theta$ and $e_i\not \in \spn( \{e \in S_{i-1}'\mid w_1(e) \geq \theta\})$ and it equals $0$ otherwise. Therefore, by the definition of $w_1^*(\cdot)$, the gain $g(\cdot)$ and $w_1(e_i) = w_1^*(e_i) + g(e_i)$ in \cref{mlr2} we have
    \begin{align*}
        w_1(T'_1) & = \sum_{i=1}^\ell \left[ w_1(e_i)  - \max\left( \{0\} \cup \lbrace \theta: e_i \in \spn\left(\lbrace f\in S'_{i-1}\ |\ w_i(f)\geq \theta\rbrace\right)\rbrace \right)\right]
         \geq \sum_{i=1}^\ell g(e_i) = g(S')\,.
    \end{align*}
    The inequality holds because $S'_{i-1}$ is a subset of the set $S$ at the time when \cref{mlr2} considers element $e_i$. Moreover, if $S' = S$, then $S'_{i-1}$ equals the set $S$ at that point and so we then have $w^*_1(e_i)= \max\left( \{0\} \cup \lbrace \theta: e_i \in \spn\left(\lbrace f\in S'_{i-1}\ |\ w_i(f)\geq \theta\rbrace\right)\rbrace \right)$, which implies that the above inequality holds with equality in that case. We can thus also conclude  that a maximum weight independent set $T_1 \subseteq S$ satisfies $w_1(T_1) = g(S)$. Finally, we can observe that $T_1$ is also a maximum weight independent set over the whole ground set since  we have  $\rank(\{e\in S\mid w_1(e) \geq \theta\}) = \rank(\{e\in E\mid w_1(e) \geq \theta\})$ for every $\theta > 0$, which holds because, by the extension of $w_1$, an element $e\not \in S$ satisfies $e\in \spn(\{f\in S: w_1(f) \geq w_1(e)\})$.
\end{proof}

We can now relate the gain of the elements in $S$ with the weight of an optimal solution.

\begin{lemma}\label{gee}
    Let $S$ be the subset generated by \cref{mlr2}.  Then $g(S) \geq w(S^*)/2$.
\end{lemma}

\begin{proof}
    We first observe that $w_1(e) + w_2(e) \geq w(e)$ for every element $e\in E$. Indeed, for an  element $e\in S$, we have by definition $w(e)=g(e)+w_1^*(e)+w_2^*(e)$, and $w_i(e)=g(e)+w_i^*(e)$, so  $w_1(e)+ w_2(e)=2g(e)+w_1^*(e)+w_2^*(e)=w(e)+g(e)>w(e)$. In the other case, when $e \not \in S$ then $w_1^*(e)+w_2^*(e)\geq w(e)$, and $w_i(e)=w_i^*(e)$, so automatically, $w_1(e)+w_2(e)\geq w(e)$.
    
   The above implies that   $w_1(S^*)+w_2(S^*)\geq w(S^*)$. On the other hand, by \cref{goodsubset}, we have $w_i(T_i)\geq w_i(S^*)$ (since $T_i$ is a max weight independent set in $M_i$ with respect to $w_i$) and $w_i(T_i)=g(S)$, thus $g(S)\geq w_i(S^*)$ for $i=1,2$.
\end{proof}

We finish the proof of \cref{k2} by proving that there is a $T \subseteq S$ independent in both $M_1$ and $M_2$ such that $w(T) \geq g(S)$. As described in \cref{sec:counterexample}, we cannot select $T$ using the greedy method. Instead, we select $T$ using the concept of kernels studied in~\cite{Fleiner}.

\begin{lemma}\label{lemma:atleastgain}
    Let $S$ be the subset generated by \cref{mlr2}. Then there exists a subset $T\subseteq S$ independent in $M_1$ and in $M_2$ such that $w(T) \geq g(S)$.
\end{lemma}
\begin{proof}
    Consider one of the matroids $M_i$ with $i\in \{1,2\}$ and define a total order $<_i$ on $E$ such that $e<_i f$ if $w_i(e) > w_i(f)$ or if $w_i(e) = w_i(f)$ and $e$ appeared later in the stream than $f$. The pair $(M_i, <_i)$  is known as an ordered matroid. We further say that  a subset $E'$ of $E$ dominates element $e$ of $E$ if $e\in E'$ or there is a subset $C_e\subseteq E'$ such that $e\in \spn(C_e)$ and $c<e$ for all elements $c$ of $C_e$. The set of elements dominated by $E'$ is denoted by $D_{M_i}(E')$. Note that if $E'$ is an independent set, then the greedy algorithm that considers the elements of $D_{M_i}(E')$ in the order $<_i$ selects exactly the elements $E'$.
    
    Theorem~2 in~\cite{Fleiner} says that for two ordered matroids $(M_1, <_1), (M_2, <_2)$ there always is a set $K\subseteq E$, which is referred to as a $M_1M_2$-kernel, such that
    \begin{itemize}
        \item $K$ is independent in both  $M_1$ and in $M_2$; and
        \item $D_{M_1}(K) \cup D_{M_2}(K) = E$.
    \end{itemize}

    We use the above result on $M_1$ and $M_2$ restricted to the elements in $S$. Specifically we select $T\subseteq S$ to be the kernel such that $D_{M_1}(T) \cup D_{M_2}(T) = S$. Let $S_1 = D_{M_1}(T)$ and $S_2 = D_{M_2}(T)$. 
     By \cref{goodsubset}, there exists a set $T'\subseteq S_1$ independent in $M_1$ such that $w_1(T')\geq g(S_1)$. As noted above, the greedy algorithm that considers the element of $S_1$ in the order $<_i$ (decreasing weights) selects exactly the elements in $T$. It follows by the optimality of the greedy algorithm for matroids that $T$ is optimal for $S_1$ in $M_1$ with weight function $w_1$, which in turn implies $w_1(T)\geq g(S_1)$. In the same way, we also have $w_2(T)\geq g(S_2)$. By definition, for any $e\in S$, we have $w(e)=w_1(e)+w_2(e)-g(e)$. Together, we have $w(T)=w_1(T)+w_2(T)-g(T)\geq g(S_1)+g(S_2)-g(T)$. As elements from $T$ are in both $S_1$ and $S_2$, and all other elements are in at least one of both sets, we have $g(S_1)+g(S_2)\geq g(S)+g(T)$, and thus $w(T)\geq g(S)$. 
\end{proof}

\section{Making the Algorithm Memory Efficient}
\label{sec:algmemefficient}

We now modify \cref{mlr2} to only select elements with a significant gain, parametrized by $\alpha>1$, and delete elements if we have too many in memory, parametrized by a real number $y$. If $\alpha$ is close enough to 1 and $y$ is large enough, then \cref{ssmlr} is very close to \cref{mlr2}, and allows for a similar analysis. This method is very similar to the one used in \cite{Paz} and \cite{Wajc}, but our analysis is quite different. 

More precisely, we take an element $e$ only if $w(e)>\alpha(w_1^*(e)+w_2^*(e))$ instead of $w(e)>w_1^*(e)+w_2^*(e)$, and we delete elements if the ratio between two $g$ weights becomes larger than $y$ ($\frac{g(e)}{g(e')}>y$). For technical purposes, we also need to keep independent sets $T_1$ and $T_2$ which maximize the weight functions $w_1$ and $w_2$ respectively. If an element with small $g$ weight is in $T_1$ or $T_2$, we do not delete it, as this would modify the $w_i$-weights and selection of coming elements. We show that this algorithm is a semi-streaming algorithm with an approximation guarantee of $(2+\epsilon)$ for an appropriate selection of the parameters (see \cref{streaming} for the space requirement  and \cref{mainresult} for the approximation guarantee).

\begin{lemma}\label{nfty}
Let $S$ be the subset generated by \cref{ssmlr} with $\alpha\geq 1$ and $y=\infty$. Then $w(S^*)\leq 2\alpha g(S)$.
\end{lemma}

\begin{proof}
We define $w_\alpha:E\rightarrow \mathbb{R}$ by $w_\alpha(e)=w(e)$ if $e\in S$ and $w_\alpha(e)=\frac{w(e)}{\alpha}$ otherwise. By construction, \cref{ssmlr} and \cref{mlr2} give the same set $S$, and the same weight function $g$ for this modified weight function. By \cref{gee}, $w_\alpha(S^*)\leq 2g(S)$. On the other hand, $w(S^*)\leq \alpha w_\alpha(S^*)$.
\end{proof}

\begin{lemma}\label{streaming}
Let $S$ be the subset generated generated by \cref{ssmlr} with $\alpha=1+\epsilon$ and $y=\frac{\min(r_1,r_2)}{\epsilon^2}$ and $S^*$ be a maximum weight independent set, where $r_1$ and $r_2$ are the ranks of $M_1$ and $M_2$ respectively. Then $w(S^*)\leq 2(1+2\epsilon+o(\epsilon))g(S)$. Furthermore, at any point of time, the size of $S$ is at most $r_1+r_2+ \min(r_1,r_2)\log_\alpha(\frac{y}{\epsilon})$.
\end{lemma}

\begin{proof}

We first prove that the generated set $S$ satisfies $w(S^*)\leq 2(1+2\epsilon+o(\epsilon))g(S)$ and we then verify the space requirement of the algorithm, i.e., that it is a semi-streaming algorithm.

Let us call $S'$ the set of elements selected by \cref{ssmlr}, including the elements deleted later. By \cref{nfty}, we have $2\alpha g(S')\geq w(S^*)$, so all we have to prove is that $g(S')-g(S)\leq \epsilon g(S)$. We set $i\in \lbrace 1,2\rbrace$ to be the index of the matroid with smaller rank.

In our analysis, it will be convenient to think that the algorithm maintains the maximum weight independent set $T_i$ of $M_i$ throughout the stream. We have,  
at the arrival of an element $e$ that is added to $S$, that the set $T_i$ is updated as follows. If $T_i \cup \{e\}\in I_i$ then $e$ is simply added to $T_i$. Otherwise,  before updating $T_i$, there is an element $e^*\in T_i$ such that $w_i(e
^*)=w_i^*(e)$ and $T_i\setminus\lbrace e
^*\rbrace\cup \lbrace e\rbrace$ is maximum weight independent set in $M_i$ with respect to $w_i$. Thus we can speak of elements which are \emph{replaced} be another element in $T_i$. By construction, if $e$ replaces $f$ in $T_i$, then $w_i(e)>\alpha w_i(f)$.

We can now divide the elements of $S'$ into stacks in the following way: If $e$ replaces an element $f$ in $T_i$, then we add $e$ on top of the stack containing $f$, otherwise we create a new stack containing only $e$. At the end of the stream, each element $e\in T_i$ is in a different stack, and each stack contains exactly one element of $T_i$, so let us call $S_e'$ the stack containing $e$ whenever $e\in T_i$. We define $S_e$ to be the restriction of $S_e'$ to $S$. In particular, each element from $S'$ is in exactly one $S_e'$ stack, and each element from $S$ is in exactly one $S_e$ stack. For each stack $S_e'$, we set $e_{del}(S_e')$ to by the highest element of $S'_e$ which was removed from $S$. By construction, $g(S_e')-g(S_e)\leq w_i(e_{del}(S_e'))$. On the other hand, $w_i(f)<\frac{1}{\epsilon}g(f)$ for any element $f\in S'$ (otherwise we would not have selected it), so $g(S_e')-g(S_e)<\frac{1}{\epsilon}g(e_{del}(S_e'))$. As $e_{del}(S_e')$ was removed from $S$, we have $g(e_{del}(S_e'))<\frac{g_{max}}{y}$. As there are exactly $r_i$ stacks, we get $g(S')-g(S)<r_i\frac{g_{max}\epsilon
^2}{r_i\epsilon}=\epsilon g_{max}\leq \epsilon g(S)$.

We now have to prove that the algorithm fits the semi-streaming criteria. In fact, the size of $S$ never exceeds $r_1+r_2+r_i \log_\alpha(\frac{y}{\epsilon})$. By the pigeonhole principle, if $S$ has at least $r_i \log_\alpha(\frac{y}{\epsilon})$ elements, then there is at least one stack $S_e$ which has at least $\log_\alpha(\frac{y}{\epsilon})$ elements. By construction, the $w_i$ weight increases by a factor of at least $\alpha$ each time we add an element on the same stack, so the $w_i$ weight difference between the lowest and highest element on the biggest stack would be at least $\frac{y}{\epsilon}$. As $w_i(f)<\frac{1}{\epsilon}g(f)$, the $g$ weight difference would be at least $y$, and we would remove the lowest element, unless it was in $T_1$ or $T_2$.


\end{proof}

\begin{theorem}\label{mainresult}
Let $S$ be the subset generated by running \cref{ssmlr} with $\alpha=1+\epsilon$ and $y=\frac{\min(r_1,r_2)}{\epsilon^2}$. Then there exists a subset $T\subseteq S$ independent in $M_1$ and in $M_2$ such that $w(T) \geq g(S)$. Furthermore, $T$ is a $2(1+2\epsilon+o(\epsilon))$-approximation for the intersection of two matroids. 
\end{theorem}

\begin{proof}
Let $S^*$ be a maximum weight independent set. By \cref{streaming}, we have $2(1+2\epsilon+o(\epsilon)g(S)\geq w(S^*)$. By \cref{lemma:atleastgain} we can select an independent set $T$ with $w(T)\geq g(S)$ if the algorithm does not delete elements. Let $S'$ be the set of elements selected by \cref{ssmlr}, including the elements deleted later. As long as we do not delete elements from $T_1$ or $T_2$, \cref{mlr2} restricted to $S'$ will select the same elements, with the same weights, so we can consider $S'$ to be generated by \cref{mlr2}. We now observe that all the arguments used in \cref{lemma:atleastgain} also work for a subset of $S'$, in particular, it is also true for $S$ that we can find an independent set $T\subseteq S$ such that $w(T)\geq g(S)$.
\end{proof}

\begin{remark}
Algorithm \ref{ssmlr} is not the most efficient possible in terms of memory, but is aimed to be simpler instead. Using the notion of stacks introduced in the proof of \cref{streaming}, it is possible to modify the algorithm and reduce the memory requirement by a factor $\log(\min(\rank(M_1),\rank(M_2)))$.
\end{remark}

\begin{remark}
    The techniques of this section can also be used in the case when the ranks of the matroids are unknown. Specifically, the algorithm can maintain the stacks created in the proof of \cref{streaming} and allow for an error $\epsilon$ in the first two stacks created, an error of $\epsilon/2$ in the next $4$ stacks, and in general an error of $\epsilon/2^i$ in the next $2^i$ stacks.
\end{remark}

\begin{remark}
It is easy to construct examples where the set $S$ only contains a $2\alpha$-approximation (for example with bipartite graphs), so up to a factor $\epsilon$ our analysis is tight.
\end{remark}

\begin{algorithm} 
\caption{Semi-streaming adaptation of \cref{mlr2}} 
\label{ssmlr} 
\begin{algorithmic} 
    \REQUIRE A stream of the elements and $2$ matroids (which we call $M_1,M_2$) on the same ground set $E$, a real number $\alpha>1$ and a real number $y$.
    \ENSURE A set $X\subseteq E$ that is independent in both matroids.
    \STATE Whenever we write an assignment of a variable with subscript $i$, it means we do it for $i=1,2$.
    \STATE $S \leftarrow \emptyset$
    \FOR{element $e$ in the stream}
    \STATE calculate $w^*_i(e) =\max\left( \{0\} \cup \lbrace \theta: e \in \spn_{M_i}\left(\lbrace f\in S\ |\ w_i(f)\geq \theta\rbrace\right)\rbrace \right)$.
    \IF{$w(e)>\alpha(w^*_1(e)+w^*_2(e))$}
    \STATE $g(e)\leftarrow w(e)- w^*_1(e)-w^*_2(e)$
    \STATE $S\leftarrow S\cup\lbrace e\rbrace$
    \STATE $w_i(e)\leftarrow g(e)+w^*_i(e)$
    \STATE Let $T_i$ be a maximum weight independent set of $M_i$ with respect to $w_i$.
    \STATE Let $g_{max}=\underset{e\in S}{\max}g(e)$
    \STATE Remove all elements $e'\in S$, such that $y\cdot g(e')<g_{max}$ and $e'\notin T_1\cup T_2$ from S.
    \ENDIF
    \ENDFOR
    \RETURN a maximum weight set $T\subseteq S$ that is independent in $M_1$ and $M_2$
    
\end{algorithmic}
\end{algorithm}

\input{submodular_extension}

\section{More than two matroids}\label{morethantwo}

We can easily extend \cref{ssmlr} to the intersection of $k$ matroids (see \cref{ssmlrk} for details). Most results remain true, in particular, we can have $kg(S)\geq(1+\epsilon)w(S^*)$ by carefully selecting $\alpha$ and $y$. The only part which does not work is the selection of the independent set from $S$. Indeed, matroid kernels are very specific to two matroids. We now prove that a similar approach fails, by proving that the logical generalization of kernels to 3 matroids is wrong, and that a counter-example can arise from \cref{ssmlrk}. Thus, any attempt to find a $k+\epsilon$ approximation using our techniques must bring some fundamentally new idea. Still, we conjecture that the generated set $S$ contains such an approximation.

\begin{proposition}\label{3bad}
There exists a set $S$ and 3 matroids $(S,I_1),(S,I_2),(S,I_3)$ such that there does not exist a set $T\subseteq S$ such that $S=D_{M_1}(T)\cup D_{M_2}(T)\cup D_{M_3}(T)$ (see \cref{lemma:atleastgain} for a definition of $D_{M_i}(T)$) and $T$ is independent in $M_1,M_2$ and $M_3$ where $<_i$ is given by $w_i$ generated by \cref{ssmlrk} (for $\alpha$ sufficiently small).
\end{proposition}

\begin{proof}
We set $S=\lbrace a,x,y,z,b\rbrace$, which are given in this order to \cref{ssmlrk}. We now define $I_1,I_2,I_3$ in the following way. A set of 2 elements is in $I_i$ if and only if:

-In $I_1$ if it is not $\lbrace a,x\rbrace$

-In $I_2$ if it is not $\lbrace a,y\rbrace$

-In $I_3$ if it is not $\lbrace a,z\rbrace$

A set of 3 elements is in $I_i$ if and only if each of its subsets of 2 elements is in $I_i$ and:

-In $I_1$ if it contains $z$

-In $I_2$ if it contains $x$

-In $I_3$ if it contains $y$

A set of 4 elements is not in $I_i$.

Let us verify that these constraints correspond to matroids. As the problem is symmetrical, it is sufficient to verify that $M_1$ is a matroid. The 3 element independent sets in $M_1$ are exactly $\lbrace y,z,b\rbrace,\lbrace x,z,b \rbrace,\lbrace x,y,z\rbrace,\lbrace a,z,b \rbrace$ and $\lbrace a,y,z\rbrace$. Now we consider $X,Y \in I_1$ with $|X|<|Y|$. We should find $e\in Y\setminus X$ such that $X\cup \lbrace e\rbrace\in I_1$. If $X=\emptyset$, take any element from $Y$. If $X$ is a singleton, then there are two cases: either it is one of $X\subseteq\lbrace a,x\rbrace$, or it is not. In any case, $Y$ contains at most one element from $\lbrace a,x \rbrace$. As it contains at least two elements, $Y$ has to contain an element from $\lbrace y,z,b\rbrace$. In the first case, we can add any of these to $X$ to get an independent set. In the second case, $X\subseteq\lbrace y,z,b\rbrace$, so we can add any element to $X$ and it will remain independent, so just pick any element from $Y\setminus X$. If $X$ contains two elements, then $Y$ is one of the sets from the list above. In particular, it contains $z$. If $z\notin X$, then we can add $z$ to $X$. Otherwise, either $X\subseteq \lbrace y,z,b\rbrace$, in which case we can add any element, or $X$ is $\lbrace a,z \rbrace$ or $\lbrace x,z\rbrace$. In either case, $Y$ must contain an element from $\lbrace y,b \rbrace$, which we can add to $X$.

We now set the weights $w(a)=1,w(x)=w(y)=w(z)=3$ and $w(b)=8$ and run \cref{ssmlrk}.
\begin{itemize}
    \item Element $a$ has weight 1, and $\lbrace a\rbrace$ is independent in $M_1,M_2$ and $M_3$, so we set $w_1(a)=w_2(a)=w_3(a)=g(a)=1$ and $a$ is added to $S$.
    \item Element $x$ is spanned by $a$ in $M_1$, and not spanned by any element in $M_2$ and $M_3$, so we get $g(x)=w(x)-w_1^*(x)-w_2^*(x)-w_3^*(x)=3-1-0-0=2$. As $2>0$, we add $x$ to $S$. We also set $w_1(x)=3$ and $w_2(x)=w_3(x)=2$.
    \item Element $y$ and $z$ are very similar to $x$.
    \item Element $b$ is spanned in all three matroids by the elements of $w_i$ weight at least 2. On the other hand, $b$ is not spanned in any matroid by the elements of $w_i$ weight strictly bigger than 2, so $w_i^*(b)=2$ for $i=1,2,3$, thus $g(b)=8-2-2-2=2$ and $w_i(b)=2+2=4$ for every $i$. 
\end{itemize}

To recapitulate, we have $w_1(a)=1,w_1(x)=3,w_1(y)=w_1(z)=2,w_1(b)=4$ and the $w_2$ and $w_3$ weights are similar, with $y$ respectively $z$ being heavier.

Let us assume for a contradiction that $T$ is a solution to the problem.

$T$ must contain $b$, as it is the heaviest element in every matroid.

If $T$ contains $a$, then it cannot contain any of $x,y,z$, otherwise it would not be independent in one of the matroids, so we would have $T\subseteq\{a,b\}$. But $x$ has to be in at least one $D_{M_i}(T)$, and the set $\{x,b\}$ is independent in every matroid, and has a bigger weight than $\{a,b\}$, so $x$ would not be in $D_{M_i}(T)$. Thus $T$ cannot contain $a$.

As the problem is symmetrical for $\lbrace x,y,z\rbrace$, it is sufficient to test $T=\lbrace z,b\rbrace,T=\lbrace y,z,b\rbrace$ and $T=\lbrace x,y,z,b\rbrace$. The last two are not in $I_2$, so the only remaining possibility is $T=\lbrace z,b\rbrace$. But then $y$ is not in $D_{M_1}$ or $D_{M_3}$ because $\lbrace z,b, y\rbrace$ is independent in $M_1$ and $M_3$, and it is not in $D_{M_2}$ because $w_2(y) > w_2(z) \Leftrightarrow y<_2 z$ and $\lbrace y,b\rbrace$ is independent in $M_2$. As $y$ is not in any $D_{M_i}$, this concludes the proof.
\end{proof}

\begin{remark}
In the example of Proposition \ref{3bad}, we have $g(S)=w(a)+w(b)$, and $\lbrace a,b\rbrace$ is independent in all 3 matroids, so this does not contradict Conjecture \ref{conj}.
\end{remark}

\begin{conjecture}\label{conj}
The stack $S$ generated by \cref{mlr2} contains a $k$ approximation for any $k$.
\end{conjecture}

In the case $k=2$, this corresponds to Theorem \ref{k2}. For any $k$, one can easily find examples were $S$ does not contain more than a $k$ approximation, but we were unable to find an example were it does not contain a $k$ approximation.

\newpage

\begin{algorithm} 
\caption{Extension of \cref{ssmlr} to $k$ matroids} 

\label{ssmlrk} 
\begin{algorithmic} 
    \REQUIRE A stream of the elements and $k$ matroids (which we call $M_1,\dotsc,M_k$) on the same ground set $E$, a real number $\alpha>1$ and a real number $y$.
    \ENSURE A set $S\subseteq E$ of ``saved'' elements.
    \STATE When we write an assignment of a variable with subscript $i$, it means we do it for $i=1,\dotsc,k$.
    \STATE $S \leftarrow \emptyset$
    \FOR{element $e$ in the stream}
    \STATE calculate $w^*_i(e) =\max\left( \{0\} \cup \lbrace \theta: e \in \spn_{M_i}\left(\lbrace f\in S\ |\ w_i(f)\geq \theta\rbrace\right)\rbrace \right)$.
    \IF{$w(e)>\alpha\sum_{i=1}^k w^*_i(e)$}
    \STATE $g(e)\leftarrow w(e)-\sum_{i=1}^k w^*_i(e)$
    \STATE $S\leftarrow S\cup\lbrace e\rbrace$
    \STATE $w_i(e)\leftarrow g(e)+w^*_i(e)$
    \STATE Let $T_i$ be a maximum weight independent set of $M_i$ with respect to $w_i$.
    \STATE Let $g_{max}=\underset{e\in S}{\max}g(e)$
    \STATE Remove all elements $e'\in S$, such that $y\cdot g(e')<g_{max}$ and $e'\notin \bigcup_{i=1}^k T_i$ from S.
    \ENDIF
    \ENDFOR
    
\end{algorithmic}
\end{algorithm}

\section{Acknowledgements}
The authors thank Moran Feldman for pointing us to the recent paper~\cite{levin2020streaming}.

\bibliographystyle{amsalpha}
\bibliography{bibliography}

\end{document}

%% file: LRExampleFigure.tex
    \begin{tikzpicture}
        \draw[rounded corners= 10pt] (-1,-1.4) rectangle (12.8,2.1);
        
        \begin{scope}[xshift=0cm]
        \node[sgvertex] (a) at (0,1.5) {$0$};
        \node at (-0.4, 1.5) {$1$};
        \node[sgvertex] (b) at (1.5,1.5) {$0$};
        \node at (1.9, 1.5) {$1$};
        \node[sgvertex] (c) at (0,0) {$0$};
        \node[sgvertex] (d) at (1.5,0) {$0$};
        
        \draw (a) edge node[above] {$1$} (b);
        \node at (0.75, -0.9) {Time 1};
        \end{scope}
        \begin{scope}[xshift=3.5cm]
0       \draw[dashed, draw=gray, thick] (-1.00, -1.2)  -- (-1.00, 1.9);
        \node[sgvertex] (a) at (0,1.5) {$1$};
        \node[sgvertex] (b) at (1.5,1.5) {$1$};
        \node at (1.9, 1.5) {$2$};
        \node[sgvertex] (c) at (0,0) {$0$};
        \node at (-0.4, 0) {$1$};
        \node[sgvertex] (d) at (1.5,0) {$0$};

        \draw (a) edge (b);
        \draw (b) edge[thick] node[above] {$2$} (c);
        \node at (0.75, -0.9) {Time 2};
        \end{scope}
        \begin{scope}[xshift=7cm]
        \draw[dashed, draw=gray, thick] (-1.00, -1.2)  -- (-1.00, 1.9);
        \node[sgvertex] (a) at (0,1.5) {$1$};
        \node[sgvertex] (b) at (1.5,1.5) {$2$};
        \node[sgvertex] (c) at (0,0) {$1$};
        \node at (-0.4, 0) {$2$};
        \node[sgvertex] (d) at (1.5,0) {$0$};
        \node at (1.9, 0) {$1$};

        \draw (a) edge (b);
        \draw (b) edge (c);
        
        \draw (c) edge[thick] node[below] {$2$} (d);

        \node at (0.75, -0.9) {Time 3};
        \end{scope}
        \begin{scope}[xshift=10.5cm]
        \draw[dashed, draw=gray, thick] (-1.00, -1.2)  -- (-1.00, 1.9);
        \node[sgvertex] (a) at (0,1.5) {$1$};
        \node[sgvertex] (b) at (1.5,1.5) {$2$};
        \node[sgvertex] (c) at (0,0) {$2$};
        \node[sgvertex] (d) at (1.5,0) {$1$};
        
        \draw (a) edge (b);
        \draw (b) edge (c);
        \draw (c) edge (d);
        \draw (a) edge[thick, dashed] node[above] {$2$} (d);
        \node at (0.75, -0.9) {Time 4};
        \end{scope}

        \begin{scope}[yshift=-3.7cm]
        
        \draw[rounded corners= 10pt] (-1,-1.4) rectangle (12.8,2.1);
        
        \begin{scope}[xshift=0cm]
        \node[sgvertex] (a) at (0,1.5) {};
        \node[sgvertex] (b) at (1.5,1.5) {};
        \node[sgvertex] (c) at (0,0) {};
        \node[sgvertex] (d) at (1.5,0) {};
        
        \draw (a) edge  (b);
        \node at (0.75, -0.9) {Time 1};
        \end{scope}
        \begin{scope}[xshift=3.5cm]
       \draw[dashed, draw=gray, thick] (-1.00, -1.2)  -- (-1.00, 1.9);
        \node[sgvertex] (a) at (0,1.5) {};
        \node[sgvertex] (b) at (1.5,1.5) {};
        \node[sgvertex] (c) at (0,0) {};
        \node[sgvertex] (d) at (1.5,0) {};

        \draw (a) edge[] node[near end, fill=white, inner sep =1] {\small $1$} node[near start, fill=white, inner sep =1] {\small $1$} (b);
        \draw (b) edge[thick]  (c);
        \node at (0.75, -0.9) {Time 2};
        \end{scope}
        \begin{scope}[xshift=7cm]
        \draw[dashed, draw=gray, thick] (-1.00, -1.2)  -- (-1.00, 1.9);
        \node[sgvertex] (a) at (0,1.5) {};
        \node[sgvertex] (b) at (1.5,1.5) {};
        \node[sgvertex] (c) at (0,0) {};
        \node[sgvertex] (d) at (1.5,0) {};

        \draw (a) edge[] node[near end, fill=white, inner sep =1] {\small $1$} node[near start, fill=white, inner sep =1] {\small $1$}(b);
        \draw (b) edge[] node[near end, fill=white, inner sep =1] {\small $1$} node[near start, fill=white, inner sep =1] {\small $2$}(c);
        
        \draw (c) edge[thick]  (d);

        \node at (0.75, -0.9) {Time 3};
        \end{scope}
        \begin{scope}[xshift=10.5cm]
        \draw[dashed, draw=gray, thick] (-1.00, -1.2)  -- (-1.00, 1.9);
        \node[sgvertex] (a) at (0,1.5) {};
        \node[sgvertex] (b) at (1.5,1.5) {};
        \node[sgvertex] (c) at (0,0) {};
        \node[sgvertex] (d) at (1.5,0) {};
        
        \draw (a) edge[] node[near end, fill=white, inner sep =1] {\small $1$} node[near start, fill=white, inner sep =1] {\small $1$}(b);
        \draw (b) edge[] node[near end, fill=white, inner sep =1] {\small $1$} node[near start, fill=white, inner sep =1] {\small $2$}(c);
        \draw (c) edge[] node[near end, fill=white, inner sep =1] {\small $1$} node[near start, fill=white, inner sep =1] {\small $2$} (d);
        \draw (a) edge[thick, dashed]  (d);
        \node at (0.75, -0.9) {Time 4};
        \end{scope}
        \end{scope}
    \end{tikzpicture}

%% file: submodular_extension.tex
\section{Extension to submodular functions}
\label{submodular_section}
In this section, we consider the problem of submodular matroid intersection in the semi-streaming model. We first give the definition of a submodular function and then formally define our problem. 

\begin{definition}[Submodular function]
A set function $f: 2^E \rightarrow \mathbb{R}$ is submodular if it satisfies that for any two sets $A,B \subseteq E$, $f(A)+f(B)\geq f(A \cup B) + f(A \cap B)$. For any two sets $A,B \subseteq E $, let $f(A \mid B) := f(A \cup B) - f(B)$. For any element $e$ and set $A \subseteq E$, let $f(e \mid A) := f(A \cup \{ e\} \mid A)$. Now, an equivalent and more intuitive definition for $f$ to be submodular is that for any two sets $A \subseteq B \subseteq E$, and $e \in E \setminus B$, it holds that $f(e \mid A) \geq f(e \mid B)$. The function $f$ is called monotone if for any element $e \in E$ and set $A \subseteq E$, it holds that $f(e\mid A) \geq 0$.
\end{definition}

Given the above definition, we can formally define our problem now. Here, we are given an oracle access to two matroids $M_1 = (E, I_1), M_2 = (E, I_2)$ on a common ground set $E$ and an oracle access to non-negative submodular function $f: 2^E \rightarrow \mathbb{R}_{\geq 0}$ on the powerset of the elements of the ground set.  
The goal is to find a subset $X \subseteq E$ that is independent in both matroids, i.e., $X\in I_1$ and $X\in I_2$, and whose weight $f(X)$ is maximized.

Our Algorithm \ref{submodular_alg} is a straightforward generalization of Algorithm \ref{mlr2} and Algorithm 1 of \cite{levin2020streaming}. Since, the weight of an element $e$ now depends on the underlying set it would be added to, we (naturally) define the weight of $e$ to be the additional value $e$ provides after adding it to set $S$, i.e. $w(e)=f(e \mid S)$. If $e$ provides $S$ a good enough value, i.e, $f(e \mid S) \geq \alpha(w_1^*(e)+w_2^*(e))$, we add it to set $S$ but with a probability $q$ now. This probability $q$ is the most important difference between Algorithm \ref{ssmlr} and Algorithm \ref{submodular_alg}. This is a trick that we borrow from the Algorithm 1 of \cite{levin2020streaming} which is useful when $f$ is non-monotone because of the following Lemma 2.2 of \cite{buchbinder2014submodular}.  

\begin{lemma}[Lemma 2.2  in \cite{buchbinder2014submodular}] \label{buchbinder14}
Let $h: 2^E \rightarrow \mathbb{R}_{\geq 0}$ be a non-negative submodular function, and let $S$ be a a random subset of $E$ containing every element of $M$ with probability at most $q$(not necessarily independently), then $E[h(S)] \geq (1-q)h(\emptyset)$. 
\end{lemma}

In our proof, we can relate the weight of the set that we pick and the value $f(S^* \cup S_f)$ where $S_f$ denotes the elements in the stack when algorithm stops and $S^*$ denotes the set of optimum elements. If the function $f$ is monotone, this is sufficient as $f(S^* \cup S_f) \geq f(S^*)$. This, however, is not true if function $f$ is non-monotone. Here, one can use the Lemma \ref{buchbinder14} with the function $h(T)=f(T \cup S^*)$. This enables us to conclude that $\E[f(S^* \cup S_f)] \geq (1-q)f(S^*).$   

\begin{algorithm} 
\caption{Extension of Algorithm \ref{ssmlr} to submodular functions} 
\label{submodular_alg} 
\begin{algorithmic} 
    \REQUIRE A stream of the elements and $2$ matroids (which we call $M_1,M_2$) on the same ground set $E$, a submodular function $f: 2^E \mapsto \mathbb{R}$, a real number $\alpha \geq 1$, a real number $q$ such that $0\leq q \leq 1$ and a real number $y$.
    \ENSURE A set $X\subseteq E$ that is independent in both matroids.
    \STATE Whenever we write an assignment of a variable with subscript $i$, it means we do it for $i=1,2$.
    \STATE $S \leftarrow \emptyset$
    \FOR{element $e$ in the stream}
    \STATE calculate $w^*_i(e) =\max\left( \{0\} \cup \lbrace \theta: e \in \spn_{M_i}\left(\lbrace f\in S\ |\ w_i(f)\geq \theta\rbrace\right)\rbrace \right)$.
    \IF{$f(e \mid S)>\alpha(w^*_1(e)+w^*_2(e))$}
    \STATE \textbf{with} probability $1-q$, \textbf{continue}; \hfill \COMMENT{//skip $e$ with probability $1-q$.}
    \STATE $g(e)\leftarrow f(e \mid S)- w^*_1(e)-w^*_2(e)$
    \STATE $S\leftarrow S\cup\lbrace e\rbrace$
    \STATE $w_i(e)\leftarrow g(e)+w^*_i(e)$
    \STATE Let $T_i$ be a maximum weight independent set of $M_i$ with respect to $w_i$.
    \STATE Let $g_{max}=\underset{e\in S}{\max}g(e)$
    \STATE Remove all elements $e'\in S$, such that $y\cdot g(e')<g_{max}$ and $e'\notin T_1\cup T_2$ from S.
    \ENDIF
    \ENDFOR
    \RETURN a maximum weight set $T\subseteq S$ that is independent in $M_1$ and $M_2$
    
\end{algorithmic}
\end{algorithm}

\subsection{Analysis of Algorithm \ref{submodular_alg}}
We extend the analysis of Section \ref{sec:algmemefficient} by using ideas from \cite{levin2020streaming} to analyze our algorithm. 
Before going into the technical details, we give a brief overview of our analysis. For sake of intuition, we assume that the Algorithm \ref{submodular_alg} does not delete elements and also does not skip elements with probability $1-q$. Then, due to the fact that the weight of an element $e$ is the additional value it provides to the current set $S$, one can relate the weight of the independent set picked with the weight of the optimal solution given the set $S_f$ i.e., $f(S^* \mid S_f)$ by basically using the analysis of the previous section. However, this is not enough as the weight of the optimal solution is $f(S^*)$. But, we can still relate the gain of $S_f$ to $f(S_f)$ similar to $\cite{levin2020streaming}$ which helps us relate $f(S^* \cup S_f)$ and weight of our solution. In order to extend it to the case when elements are skipped with probability $1-q$, we show the above to hold in expectation similar to $\cite{levin2020streaming}$ which is helpful for dealing with non-monotone functions because of Lemma \ref{buchbinder14}. Finally, we remark that one can use an analysis similar to Section \ref{sec:algmemefficient}, to show that the effect of deleting elements does not affect the weight of solution by a lot. 

Let $S_f$ denote the set $S$ generated when the algorithm stops and $S_f'$ denote the union of $S_f$ and the elements that were deleted by the algorithm. For sake of analysis, we define the weight function $w:E \rightarrow \mathbb{R}$ of an element $e$ to be the additional value it provided to the set $S$ when it appeared in the stream, i.e., $w(e)=f(e \mid S)$. Like before, we extend the definition of weight functions $w_1$ and $w_2$ for an element $e$ that is not added to $S$ as $w_i(e)=w_i^*(e)$ for $i \in \{1,2\}$. 
We note here that all the functions defined above are random variables which depend on the internal randomness of the algorithm. Unless we explicitly mention it, we generally talk about statements with respect to any fixed realization of the internal random choices of the algorithm.

In our analysis, we will prove properties about our algorithm that are already proven for Algorithms \ref{mlr2} and \ref{ssmlr} in the previous sections. Our proof strategy will be simply running Algorithm \ref{mlr2} or \ref{ssmlr} with the appropriate weight function which will mimick running our original algorithm. Hence, we will prove these statements in a black-box fashion. A weight function that we will use repeatedly in our proofs is $w': E \rightarrow \mathbb{R}_{\geq 0}$ where $w'(e)=w(e)$ if $e \in S_f'$, otherwise $w'(e)=0$. This basically has the effect of discarding elements not in $S_f'$ i.e, elements that were never picked by the algorithm either  because they did not provide a good enough value or because they did but were still skipped.

\begin{lemma} \label{submodular_lem1}
Consider the set $S_f'$ which is the union of $S_f$ generated by the Algorithm \ref{submodular_alg} and the elements it deletes. Then a maximum weight independent set in $M_i$ for $i \in \{1,2\}$ over the whole ground set $E$ can be selected to be a subset of $S_f'$, i.e. $T_i \subseteq S_f'$ and it satisfies $w_i(T_i)=g(S_f')$.
\end{lemma}
\begin{proof}
Consider running the Algorithm \ref{mlr2} with weight function $w'$. Notice that doing this generates a stack containing exactly the elements in the set $S_f'$ and exactly the same functions $w_1,w_2$ and $g$. Now by applying Lemma \ref{goodsubset}, we get our result.     
\end{proof}

We prove the following lemma similar to \cite{levin2020streaming} which relates the gain of elements in $S_f'$ to the weight of the optimal solution given the set $S_f'$ i.e, $f(S^* \mid S'_f)$. Notice that the below lemma holds only in expectation for $q \neq 1$. 

\begin{lemma}\label{submodular_lem2}
Denote the set $S_f'$ which is the union of $S_f$ generated by the Algorithm \ref{submodular_alg} with $q \in \{1/(2\alpha+1), 1 \}$ and the elements it deletes. Then, $\E [f(S^* \mid S'_f)]\leq 2\alpha \E[g(S_f')]$.
\end{lemma}

\begin{proof}
We first prove the lemma for $q=1$ as the proof is easier than that for $q=1/(2\alpha+1)$. Consider running the Algorithm \ref{mlr2} with weight function $w'':E \rightarrow \mathbb{R}_{\geq0}$ defined as follows. If $e \in S_f'$, $w''(e)=w(e)$, else $w''(e)=w(e)/\alpha$. Notice that doing this generates a stack containing exactly the elements in the set $S_f'$ and exactly the same functions $w_1,w_2$ and $g$. Now by applying Lemma \ref{gee}, we get that $w(S^*)\leq 2\alpha g(S_f')$. By submodularity, we get $f(S^* \mid S_f')\leq 2\alpha g(S_f')$. 

Now, we prove the lemma for $q=1/(2\alpha+1)$. We first define $\lambda: E \rightarrow \mathbb{R}$ for an element $e \in E$ as $\lambda(e)=f(e \mid S_f')$. Notice that, by submodularity of $f$ and definition of $\lambda$, we have $f(S^* \mid S_f') \leq \lambda(S^*)$. Hence, it suffices to prove $\E [\lambda(S^*)]\leq 2\alpha \E[g(S_f')]$. We prove this below. 

Let the event that the element $e \in E$ does not give us a good enough value i.e, it satisfies $\alpha(w_1^*(e)+w_2^*(e)) \geq w(e)$ be $R_e$. We have two cases to consider now.  
\begin{enumerate}
    \item  The first is when  $R_e$ is true. Then, for any fixed choice of randomness of the algorithm for which $R_e$ is true, we argue as follows. By definition, $w_i(e)=w_i^*(e)$. Hence, $\alpha(w_1(e)+w_2(e)) \geq w(e)$. Also, $w(e)=f(e \mid S)$ where $S$ is the stack when $e$ appeared in the stream. As $S \subseteq S_f'$, by submodularity and definition of $\lambda$, we get that $w(e) \geq \lambda(e)$. Hence, we also get that $\alpha\E[w_1(e)+w_2(e)| R_e] \geq \E[\lambda(e)|R_e]$.

    \item The second is when $R_e$ is false. Then, for any fixed choice of randomness of the algorithm for which $R_e$ is false, we argue as follows. Here, $e$ is picked with probability $q$ given the set $S$ at the time $e$ appeared in the stream. If we pick $e$, then $w_1(e)+w_2(e) = g(e)+w_1^*(e)+g(e)+w_2^*(e) = 2w(e)-w_1^*(e)-w_2^*(e)$. Otherwise, if we do not pick $e$, then $w_1(e)+w_2(e)=w_1^*(e)+w_2^*(e)$. Hence, the expected value of $w_1(e)+w_2(e)$ satisfies, $$\E [w_1(e)+w_2(e)| \neg R_e, S]=2qw(e)+(1-2q)(w_1^*(e)+w_2^*(e)) \geq 2qw(e).$$ The last inequality follows as we have $q=1/(2\alpha+1) \leq1/2$. By the choice of $q$ and submodularity, we get that  $\alpha\E [w_1(e)+w_2(e)| \neg R_e, S] \geq  2q\alpha w(e) = (1-q)w(e)\geq(1-q)\lambda(e)$. By law of total expectation and conditioned on $R_e$ not taking place we get, $\alpha\E [w_1(e)+w_2(e)| \neg R_e] \geq \E[\lambda(e)|\neg R_e]$. 
\end{enumerate}

Finally by the law of total expectation and the points 1 and 2, we obtain that $\alpha\E[w_1(e)+w_2(e)] \geq \E[\lambda(e)]$ holds for any element $e \in E$. Applying this to elements of $S^*$, we get that $\alpha\E[w_1(S^*)+w_2(S^*)] \geq \E[\lambda(S^*)]$. On the other hand, by Lemma \ref{submodular_lem1}, we have $w_i(T_i)\geq w_i(S^*)$ (since $T_i$ is a max weight independent set in $M_i$ with respect to $w_i$) and $w_i(T_i)=g(S_f')$, thus $g(S_f')\geq w_i(S^*)$ for $i=1,2$. Hence, we get that $\E [\lambda(S^*)]\leq 2\alpha \E[g(S_f')]$.
\end{proof}

Since, we would like the relate the gain of elements in $S_f'$ to the optimal solution we bound the value of $f(S_f')$ in terms of the gain below similar to \cite{levin2020streaming}.

\begin{lemma}\label{submodular_lem3}
Consider the set $S_f'$ which is the union of $S_f$ generated by the Algorithm \ref{submodular_alg} and the elements it deletes. Then, $g(S_f') \geq (1-1/\alpha)f(S_f')$.  
\end{lemma}

\begin{proof}
By definition, any element $e \in S_f'$, should have satisfied $w(e) \geq \alpha(w_1^*(e)+w_2^*(e))$. Hence, $g(e)\geq w(e)-w(e)/\alpha$. Summing over all elements in $S_f'$, we get $g(S_f')\geq (1-1/\alpha)w(S_f') \geq f(S_f')$ where last inequality (not an equality as $S_f'$ also contains deleted elements) follows by definition of $w$ and submodularity of $f$. 
\end{proof}

Our algorithm only has the set $S_f$ and not $S_f'$ which also includes the deleted elements. Hence, in our next lemma, we prove that the gain of elements in these two sets is roughly the same.

\begin{lemma}\label{submodular_lem4}
Consider the set $S_f'$ which is the union of $S_f$ generated by running the Algorithm \ref{submodular_alg} with $\alpha>1$, $y=\min(r_1, r_2)/\delta^2$ for any $\delta$, such that $0 < \delta \leq \alpha-1$ and the elements it deletes. Here, $r_i$ is the rank of $M_i$ for $i \in \{1,2\}$. Then, $g(S_f') - g(S_f) \leq \delta g(S_f)$. Moreover, at any point during the execution, $S$ contains at most $r_1+r_2+ \min(r_1,r_2)\log_\alpha(\frac{y}{\alpha-1})$ elements.   
\end{lemma}
\begin{proof}
Consider running the Algorithm \ref{ssmlr} with weight function $w'$. Notice that doing this generates a stack containing exactly the elements as in the set $S_f$, exactly the same set of deleted elements and exactly the same functions $w_1,w_2$ and $g$. Moreover, this generates the exact same stacks as the Algorithm \ref{submodular_alg} at every point of execution. Now by the proof of Lemma \ref{streaming}, we get our result. 
\end{proof}

Lastly, we prove that there exists a set $T$ that is independent in both matroids and has a weight atleast the gain of the elements in $S_f$. 

\begin{lemma}\label{submodular_lem5}
 Let $S_f$ be the subset generated by Algorithm \ref{submodular_alg}. Then there exists a subset $T\subseteq S$ independent in $M_1$ and in $M_2$ such that $w(T) \geq g(S_f)$.
\end{lemma}
\begin{proof}
Consider running the Algorithm \ref{ssmlr} with weight function $w'$. Recall that for any element $e \in S_f'$, $w'(e)=w(e)$, otherwise $w'(e)=0$. Notice that doing this generates a stack containing exactly the elements as in the set $S_f$ and exactly the same functions $w_1,w_2$ and $g$. The result follows by Theorem \ref{mainresult}.  
\end{proof}

Now, we have all the lemmas to prove our main theorem which we state below. 

\begin{theorem}\label{submodular_th1}
The subset $S_f$ generated by \cref{submodular_alg} with $\alpha>1$, $q \in \{1/(2\alpha+1),1 \}$ and $y=\min(r_1, r_2)/\delta^2$ for any $\delta$, such that $0<\delta \leq \alpha-1$ contains a $(4\alpha^2-1)/(2\alpha-2) + O(\delta)$ approximation  in expectation for the intersection of two matroids with respect to a non-monotone submodular function $f$. This is optimized by taking $\alpha=1+\sqrt{3}/2$, resulting in an approximation ratio of $4+2\sqrt{3}+O(\delta) \sim 7.464$. Moreover, the same algorithm run with $q=1$ and $y=\min(r_1, r_2)/\delta^2$ is $(2\alpha + \alpha/(\alpha-1))+O(\delta)$ approximate if $f$ is monotone. This is optimized by taking $\alpha=1+1/\sqrt{2}$, which yields a $3+2\sqrt{2}+O(\delta) \sim 5.828$ approximation.
\end{theorem}

\begin{proof}
By Lemmas \ref{submodular_lem2} and \ref{submodular_lem3}, we have that $2\alpha\E [g(S_f')] \geq \E [f(S^* \mid S'_f)]$ and $g(S_f')(\alpha/(\alpha-1)) \geq f(S_f')$. Combining them, we get, $$(2\alpha + \alpha/(\alpha-1))\E [g(S_f')] \geq \E[f(S_f') + f(S^* \mid S_f')] = \E[f(S^* \cup S_f')].$$ By Lemma \ref{submodular_lem4}, we also get that $g(S_f')-g(S_f) \leq \delta g(S_f)$. This gives us that  
$$(2\alpha + \alpha/(\alpha-1))(1+\delta)\E [g(S_f)] \geq \E[f(S^* \cup S_f')].$$  

Now, by Lemma \ref{submodular_lem5}, there exists a subset $T \subseteq S_f$ independent in $M_1$ and $M_2$ such that $w(T) \geq g(S_f)$. By definition of $w$, and submodularity of $f$, we get that $f(T) \geq w(T)$. This in turn implies, $f(T) \geq g(S_f)$. This gives us that
$$(2\alpha + \alpha/(\alpha-1))(1+\delta)\E [f(T)] \geq \E [f(S^* \cup S_f')].$$ 

Notice that the above inequality also holds if $q=1$ as all the above arguments also work if $q=1$. Hence, if $f$ is monotone, we get $f(S^* \cup S_f) \geq f(S^*)$ which gives us our desired inequality by rearranging terms. However, if $f$ is non-monotone one has to work a little more which we show below.

To deal with the case when $f$ is non-monotone, we use Lemma \ref{buchbinder14} and take $h(T)=f(S^* \cup T)$ for any $T \subseteq E$ within the lemma statement, to get that $\E[f(S^* \cup S_f')] \geq (1-q)f(S^*)$ as every element of $E$ appears in $S_f'$ with probability at most $q$. Putting everything together, we get that   
$$(2\alpha + \alpha/(\alpha-1))(1+\delta)\E [f(T)] \geq (1-q)f(S^*).$$
Now, substituting the value of $q=1/(2\alpha+1)$ and rearranging terms, we get the desired inequality. 

\end{proof}

\begin{remark}
We can exactly match the approximation ratios in \cite{levin2020streaming} i.e, without the extra additive factor of $O(\delta)$ by not deleting elements. Moreover, $S$ stores at most $O(\min(r_1,r_2)\log_\alpha |E|)$ elements at any point if we assume that values of $f$ are polynomially bounded in $|E|$, an assumption that the authors in \cite{levin2020streaming} make.  
\end{remark}

%% file: arxiv.bbl
\newcommand{\etalchar}[1]{$^{#1}$}
\providecommand{\bysame}{\leavevmode\hbox to3em{\hrulefill}\thinspace}
\providecommand{\MR}{\relax\ifhmode\unskip\space\fi MR }
\providecommand{\MRhref}[2]{%
  \href{http://www.ams.org/mathscinet-getitem?mr=#1}{#2}
}
\providecommand{\href}[2]{#2}
\begin{thebibliography}{BYBFR04}

\bibitem[BFNS14]{buchbinder2014submodular}
Niv Buchbinder, Moran Feldman, Joseph Naor, and Roy Schwartz, \emph{Submodular
  maximization with cardinality constraints}, Proceedings of the twenty-fifth
  annual ACM-SIAM symposium on Discrete algorithms, SIAM, 2014, pp.~1433--1452.

\bibitem[BYBFR04]{lrsurvey}
Reuven Bar-Yehuda, Keren Bendel, Ari Freund, and Dror Rawitz, \emph{Local
  ratio: A unified framework for approximation algorithms. in memoriam: Shimon
  even 1935-2004}, ACM Computing Surveys (CSUR) \textbf{36} (2004), 422--463.

\bibitem[BYE85]{lrtheorem}
R.~Bar-Yehuda and S.~Even, \emph{A local-ratio theorem for approximating the
  weighted vertex cover problem}, Analysis and Design of Algorithms for
  Combinatorial Problems (G.~Ausiello and M.~Lucertini, eds.), North-Holland
  Mathematics Studies, vol. 109, North-Holland, 1985, pp.~27 -- 45.

\bibitem[CK13]{Chak}
Amit Chakrabarti and Sagar Kale, \emph{Submodular maximization meets streaming:
  Matchings, matroids, and more}, CoRR \textbf{abs/1309.2038} (2013).

\bibitem[CS14]{Crouch}
M.~Crouch and D.M. Stubbs, \emph{Improved streaming algorithms for weighted
  matching, via unweighted matching}, Leibniz International Proceedings in
  Informatics, LIPIcs \textbf{28} (2014), 96--104.

\bibitem[Edm79]{EDMONDS197939}
Jack Edmonds, \emph{Matroid intersection}, Discrete Optimization I, Annals of
  Discrete Mathematics, vol.~4, Elsevier, 1979, pp.~39 -- 49.

\bibitem[FKM{\etalchar{+}}04]{Feigenbaum}
Joan Feigenbaum, Sampath Kannan, Andrew McGregor, Siddharth Suri, and Jian
  Zhang, \emph{On graph problems in a semi-streaming model}, 07 2004,
  pp.~531--543.

\bibitem[Fle01]{Fleiner}
Tam{\'a}s Fleiner, \emph{A matroid generalization of the stable matching
  polytope}, vol. 2081, 06 2001, pp.~105--114.

\bibitem[GW18]{Wajc}
Mohsen Ghaffari and David Wajc, \emph{{Simplified and Space-Optimal
  Semi-Streaming (2+epsilon)-Approximate Matching}}, 2nd Symposium on
  Simplicity in Algorithms (SOSA 2019) (Dagstuhl, Germany) (Jeremy~T. Fineman
  and Michael Mitzenmacher, eds.), OpenAccess Series in Informatics (OASIcs),
  vol.~69, Schloss Dagstuhl--Leibniz-Zentrum fuer Informatik, 2018,
  pp.~13:1--13:8.

\bibitem[HKMY20]{DBLP:journals/corr/abs-2002-05477}
Chien{-}Chung Huang, Naonori Kakimura, Simon Mauras, and Yuichi Yoshida,
  \emph{Approximability of monotone submodular function maximization under
  cardinality and matroid constraints in the streaming model}, CoRR
  \textbf{abs/2002.05477} (2020).

\bibitem[LW20]{levin2020streaming}
Roie Levin and David Wajc, \emph{Streaming submodular matching meets the
  primal-dual method}, arXiv preprint arXiv:2008.10062 (2020).

\bibitem[McG05]{McGregor}
Andrew McGregor, \emph{Finding graph matchings in data streams}, Approximation,
  Randomization and Combinatorial Optimization. Algorithms and Techniques
  (Berlin, Heidelberg) (Chandra Chekuri, Klaus Jansen, Jos{\'e} D.~P. Rolim,
  and Luca Trevisan, eds.), Springer Berlin Heidelberg, 2005, pp.~170--181.

\bibitem[Mut05]{Muthu}
S.~Muthukrishnan, \emph{Data streams: algorithms and applications}, Found.
  Trends Theor. Comput. Sci. \textbf{1} (2005), no.~2, 117--236. \MR{2379507}

\bibitem[PS17]{Paz}
Ami Paz and Gregory Schwartzman, \emph{A (2+{\(\epsilon\)})-approximation for
  maximum weight matching in the semi-streaming model}, CoRR
  \textbf{abs/1702.04536} (2017).

\bibitem[Sch03]{schrijver2003combinatorial}
A.~Schrijver, \emph{Combinatorial optimization: Polyhedra and efficiency},
  Algorithms and Combinatorics, Springer, 2003.

\end{thebibliography}
